\documentclass[pdflatex,sn-mathphys-ay]{sn-jnl}

\usepackage{graphicx}%
\usepackage{multirow}%
\usepackage{amsmath,amssymb,amsfonts}%
\usepackage{amsthm}%
\usepackage{mathrsfs}%
\usepackage[title]{appendix}%
\usepackage{xcolor}%
\usepackage{textcomp}%
\usepackage{manyfoot}%
\usepackage{booktabs}%
\usepackage{algorithm}%
\usepackage{algorithmicx}%
\usepackage{algpseudocode}%
\usepackage{listings}%
\usepackage[export]{adjustbox}

\usepackage[T1]{fontenc}
\usepackage[latin9]{inputenc}
\usepackage{hyperref}
\usepackage{babel}

\theoremstyle{thmstyleone}%

\newtheorem{proposition}{Proposition}[section]
\newtheorem{lemma}{Lemma}[section]
\newtheorem{corollary}{Corollary}[section]

\theoremstyle{thmstyletwo}%
\newtheorem{remark}{Remark}%

\newtheorem{assumption}{Assumption}
\newtheorem{heuristic}{Heuristic}

\theoremstyle{thmstylethree}%
\begin{document}

\title{First analytical coverage bounds of a fully specified nested sampling algorithm}

\author*[1]{\fnm{Johannes} \sur{Buchner}}\email{johannes.buchner.acad@gmx.com}

\affil*[1]{\orgname{Max Planck Institute for Extraterrestrial Physics},
\orgaddress{\street{Giessenbachstrasse}, \city{Garching},
\postcode{85748}, \country{Germany}}}

\abstract{
Nested sampling is a Monte Carlo algorithm for posterior estimation and Bayesian model comparison. It maintains a population of $K$ live points sampled from the prior, and at each iteration discards the lowest-likelihood point and replaces it with a new sample drawn from the prior restricted to exceed the discarded likelihood. Achieving this likelihood-restricted prior sampling efficiently and reliably is the central computational challenge. For low-to-moderate dimensional problems, MLFriends is a general and robust region-based approach that constructs a proposal region by bootstrap aggregation over the current live points and rejects proposals outside this region. We present a self-contained mathematical formulation of MLFriends and derive, under a homogeneous Binomial point process model for the live points, heuristic bounds on the expected fraction of the likelihood-restricted prior not covered by the proposal region. These bounds decay as $(\frac{1}{3}Km)^{-3/2}$, where $m$ is the number of bootstrap rounds, and are negligibly small for practical parameter choices. We show heuristically that the resulting bias in the marginal likelihood estimate is negligible compared to the inherent statistical variance of a nested sampling run. While a fully rigorous treatment remains an open problem, these results provide the first analytical characterisation of a fully specified and practically implementable nested sampling algorithm, without assuming an idealised or asymptotic sampling procedure.
}

\keywords{Nested sampling, Monte Carlo algorithms}

\maketitle

\section{Introduction}

In the physical sciences, Skilling's nested sampling \citep{Skilling2004,skilling2006nested}
is a popular Monte Carlo algorithm to achieve Bayesian model comparison
and posterior distributions in practice. Sophisticated physical models
can induce posteriors that are complicated and multi-modal. Nested
sampling can successfully reconstruct these thanks to making only a few assumptions.

Nested sampling integrates a posterior over an arbitrary prior space
$\Omega$ by a one-dimensional integral transform:
$
\int_{\Omega}{\cal L}(\theta)\,d\pi(\theta)=\int_{0}^{1}{\cal L}(X)\,dX,
$
where ${\cal L}(X)$ is the inverse of $X(L_{\mathrm{min}})$, the
survival function of the likelihood-restricted prior:
$
X(L_{\mathrm{min}})=\Pr\{{\cal L}(\theta)>L_{\mathrm{min}}\}
=\int_{\{\theta:\,{\cal L}(\theta)>L_{\mathrm{min}}\}}d\pi(\theta).
$

Suppose $\theta_{1},\ldots,\theta_{K}$ are i.i.d.\ samples from the
prior $\pi$, with likelihoods ${\cal L}_{1},\ldots,{\cal L}_{K}$.
Then $X({\cal L}_{1}),\ldots,X({\cal L}_{K})$ are uniformly distributed.
Discarding the lowest-likelihood point,
$L_{\mathrm{min}}=\min_{i}\{{\cal L}_{i}\}$,
removes a fraction $1-X(L_{\mathrm{min}})$ of the prior mass, where
$1-X(L_{\mathrm{min}})\sim\mathrm{Beta}(1,K)$ with mean $1/(K+1)$.
Equivalently, the remaining prior mass is $X(L_{\mathrm{min}})\sim\mathrm{Beta}(K,1)$.
Nested sampling iteratively applies this idea. At each iteration,
the point discarded from the live points is replaced with a new sample
from the prior restricted to ${\cal L}>L_{\mathrm{min}}$. The prior mass
$X$ shrinks in each iteration by an estimated fraction $\frac{K-1}{K}$.
Finally, upon reaching some convergence criterion \citep[see e.g.,][]{2014arXiv1412.6368W,Salomone2018},
the discarded points $\theta_{i}$ are assigned the unnormalised weights
$w_{i}={\cal L}_{i}\times V_{i}$, where
$V_{i}=\frac{1}{K}\times\left(\frac{K-1}{K}\right)^{i-1}$
is the estimated prior mass discarded. The posterior is approximated
by these posterior samples with normalised weights $w_{i}/\sum w_{i}$,
and the marginal likelihood by $Z=\sum w_{i}$. The convergence of
the posterior distribution and the marginal likelihood estimator to
the truth has been established under various assumptions in \citep{Chopin2010,Schittenhelm2020,evans2007discussion,skilling2009nested}.

Above, the availability of a procedure for likelihood-restricted prior
sampling \citep[LRPS, see][for a literature review]{Buchner2021c}
has been assumed. For practical use of nested sampling,
efficient LRPS algorithms are essential. The focus of this work is
to extend the convergence proofs above to such more realistic settings.

In high-dimensional settings, `step sampling' LRPS methods achieve acceptable
sampling efficiencies with random walks (such as Markov Chain Monte
Carlo; MCMC) started at a randomly selected live point \citep[see e.g.][who use slice sampling]{Jasa2012,Handley2015}.
From a sequential Monte Carlo perspective, \citet{Salomone2018} conclude
that some remaining dependence between start and end point in step
sampling LRPS is not damaging to the convergence of nested sampling.
This agrees with empirical observations \citep[see also the simulations of][]{Salomone2018}.
\citet{Buchner2023} observed that for a given MCMC proposal,
the convergence speed is set by the product of the number of live
points and the number of MCMC steps. This is because a run with twice
as many live points requires, to reach the same likelihood threshold,
twice as many iterations. Since one discarded point is now accompanied
by another point, this eases the requirement for how far the MCMC
has to diffuse in a single step. The theoretical implication is that
nested sampling converges to the true posterior as $K\rightarrow\infty$,
even with an inefficient MCMC proposal. The practical implication
is that such convergence can be achieved by resuming a nested sampling
run with ever more live points \citep[see][for how such resuming can be implemented]{skilling2009nested,Higson2019,Speagle2020,Buchner2021c},
and the convergence may be diagnosed.

In low-dimensional settings, rejection sampling from a sub-space of
the prior is efficient \citep{Shaw2007,Feroz2008,Mukherjee2006,Buchner2014stats},
termed `region-based' LRPS. In this work, we present a novel analysis
of a region-based LRPS algorithm, MLFriends \citep{Buchner2014stats,Buchner2019c,UltraNest}.
In particular, we derive heuristic bounds on the probability that the constructed 
sub-space for rejection sampling is not a superset of the 
likelihood-restricted prior. We provide the first analytical characterisation of
MLFriends and practically useful reliability estimates for
computer implementations with finite resources.

We first introduce notation and the MLFriends algorithm
in Section~\ref{sec:Materials-and-Methods}. Our main analytical results are
developed in Section~\ref{sec:correctness}.
Section~\ref{sec:efficiency} comments on the acceptance rate, and
Section~\ref{sec:implement} gives details
for an efficient implementation for non-trivial Bayesian models.
Section~\ref{sec:Discussion} discusses the implications of our results.

\section{Preliminaries}\label{sec:Materials-and-Methods}

\subsection{Assumptions}
\label{subsec:Problem-setup}

For the analysis presented in this work, we make several assumptions.
We assume the prior probability density $\pi$ is
defined over a continuous parameter space $\Omega_{\pi}\subseteq\mathbb{R}^{d}$.
Furthermore, we assume the likelihood is free of plateaus: formally,
the pushforward of $\pi$ under ${\cal L}$ is a continuous
distribution, so that for every $L_{\mathrm{min}}$ in the range
$(0,\max_\theta{\cal L}(\theta))$ we have
$\Pr\{{\cal L}(\theta)=L_{\mathrm{min}}\}=0$
\citep[see][for the general case]{Schittenhelm2020,Fowlie2021}.

To simplify our analysis, we assume the space is parameterized such
that the prior is a standard uniform distribution. This also follows
many (but not all) current implementations of nested sampling.
To achieve this, the cumulative distribution function gives the needed reparameterization
to natural probability units. For factorized priors
$\pi(\theta)=\prod_{i=1}^{d}\pi_{i}(\theta_{i})$,
each marginal CDF provides the transformation independently:
\begin{equation}
\theta_{i}'(\theta_{i})=\int_{-\infty}^{\theta_{i}}\pi_{i}(x)\,dx.\label{eq:cdf}
\end{equation}
More generally, non-factorized priors, including correlated multivariate
Gaussians and hierarchical models, can be handled by applying the probability integral transform iteratively using conditional distributions:
\begin{equation}
\theta_{i}'(\theta_{i}|\theta_{1},\ldots,\theta_{i-1})=\int_{-\infty}^{\theta_{i}}\pi(\theta_{i}'|\theta_{1},\ldots,\theta_{i-1})\,d\theta_{i}',\label{eq:cdf-conditional}
\end{equation}
where the ordering of parameters is arbitrary. For example, a multivariate
Gaussian prior can be transformed parameter-by-parameter using its
conditional cumulative probabilities. Similarly, hierarchical models can be
accommodated by conditioning on higher-level parameters.
For this reparameterization to be valid in either case, we require each
conditional density $\pi(\theta_{i}|\theta_{1},\ldots,\theta_{i-1})>0$
almost everywhere on its support, ensuring each conditional CDF is strictly monotone
and hence invertible.
Under this transformation, the prior space becomes the unit hypercube
$\Omega_{\pi'}=[0,1]^{d}$ with volume $|\Omega_{\pi'}|=1$, and the
prior density simplifies to
$\pi'(\theta')=\mathbf{1}(\theta'\in[0,1]^{d})$.
The original parameters can be retrieved by inverting
eqs.~\ref{eq:cdf}--\ref{eq:cdf-conditional}.
For notational simplicity, we drop the distinction between $\theta_{i}'$ and $\theta$
henceforth and write
$\Omega_{\pi}=[0,1]^d$ and $\pi(\theta)=\mathbf{1}(\theta\in[0,1]^d)$.

\subsection{Rejection sampling}

The goal of any LRPS algorithm is to simulate new live points. In
the region-based class of LRPS algorithms \citep[see][for a review]{Buchner2021c},
new points are proposed from the prior within a region
$\Omega_{P}\subseteq\Omega_{\pi}$:
\begin{equation}
\theta^{*}\sim\mathrm{Uniform}(\Omega_{P}).
\end{equation}
The rejection sampling acceptance rule is then simple: if the likelihood
${\cal L}(\theta^{*})$ exceeds the current threshold $L_{\mathrm{min}}$,
the proposed point is accepted.
The efficiency is proportional to how closely $\Omega_{P}$ resembles
the restricted prior space
$\Omega_{r\pi}=\{\theta\in\Omega_{\pi}:{\cal L}(\theta)>L_{\mathrm{min}}\}$:
\begin{equation}
A=\frac{|\Omega_{P}\cap\Omega_{r\pi}|}{|\Omega_{P}|}.
\end{equation}
For the proposal to be valid (i.e.\ to allow sampling from
$\Omega_{r\pi}$), we require $\Omega_{P}\supseteq\Omega_{r\pi}$,
so that $A=|\Omega_{r\pi}|/|\Omega_{P}|$.

\subsection{The MLFriends algorithm}
\label{subsec:Bootstrapping-a-region}

The MLFriends algorithm constructs $\Omega_{P}$ from the current live
points. Let ${\cal I}=\{1,\ldots,K\}$ index the live points
$\{\theta^{1},\ldots,\theta^{K}\}$. After fixing a distance metric
$d(\cdot,\cdot)$, the neighbourhood of live point $j$ at radius $r$
is:
\begin{equation}
\Omega^{j}(r)=\{x\in\Omega_{\pi}:d(x,\theta^{j})<r\}.
\end{equation}
As detailed in Section~\ref{sec:implement}, MLFriends uses a
Mahalanobis distance, but for the analysis here the space can be
affinely transformed so that the Mahalanobis distance becomes a
standard Euclidean distance, without loss of generality.

The core of the MLFriends algorithm is a bootstrap aggregating
(bagging) procedure applied $m$ times to estimate a safe radius $r$.
The bootstrap dataset (training sample) is
a sequence of $K$ indices drawn i.i.d.\ uniformly with replacement
from $\{1,\ldots,K\}$:
\begin{equation}
{\cal I}_{\mathrm{train}}=(i_{1},\ldots,i_{K}),
\quad i_{j}\overset{\mathrm{iid}}{\sim}\mathrm{Uniform}\{1,\ldots,K\}.
\end{equation}
This sequence may contain repeated indices. The corresponding set of
unique indices is:
\begin{equation}
{\cal I}_{\mathrm{train}}^{\mathrm{uniq}}=\{i_{1},\ldots,i_{K}\}
\quad(\text{as a set, without repetition}).
\end{equation}
The validation set consists of live points whose indices do not appear
in ${\cal I}_{\mathrm{train}}^{\mathrm{uniq}}$:
\begin{equation}
{\cal I}_{\mathrm{val}}={\cal I}\setminus{\cal I}_{\mathrm{train}}^{\mathrm{uniq}}.
\end{equation}

\begin{figure}
\centering
\includegraphics[width=0.49\columnwidth,valign=t,trim={0 2.7cm 0 0},clip]{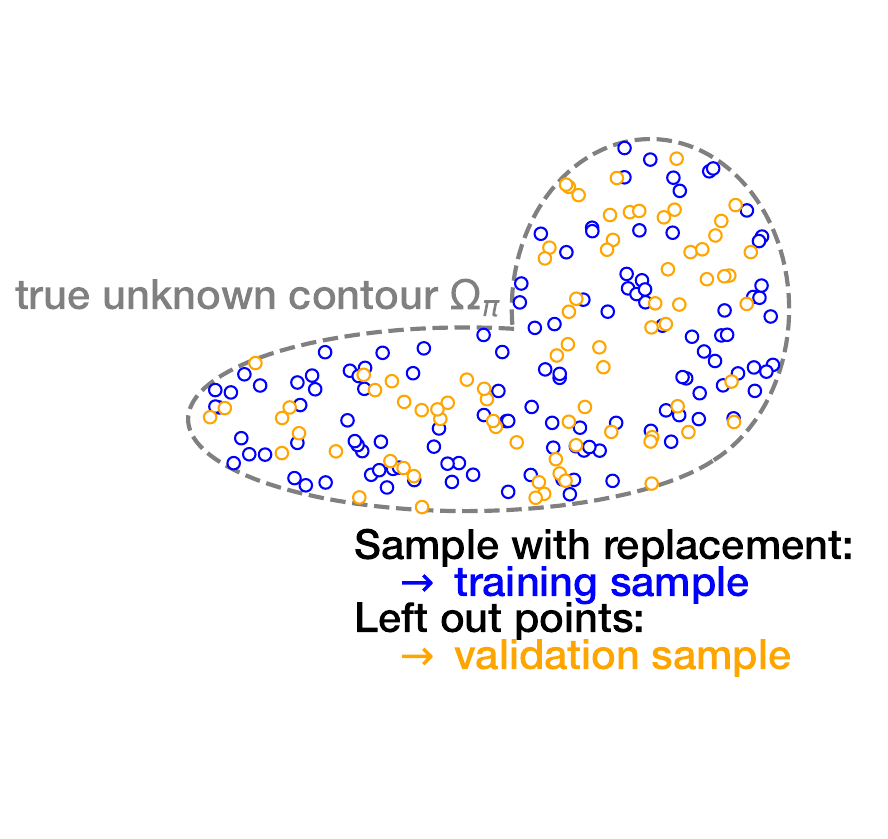}%
\includegraphics[width=0.49\columnwidth,valign=t,trim={0 2.7cm 0 0},clip]{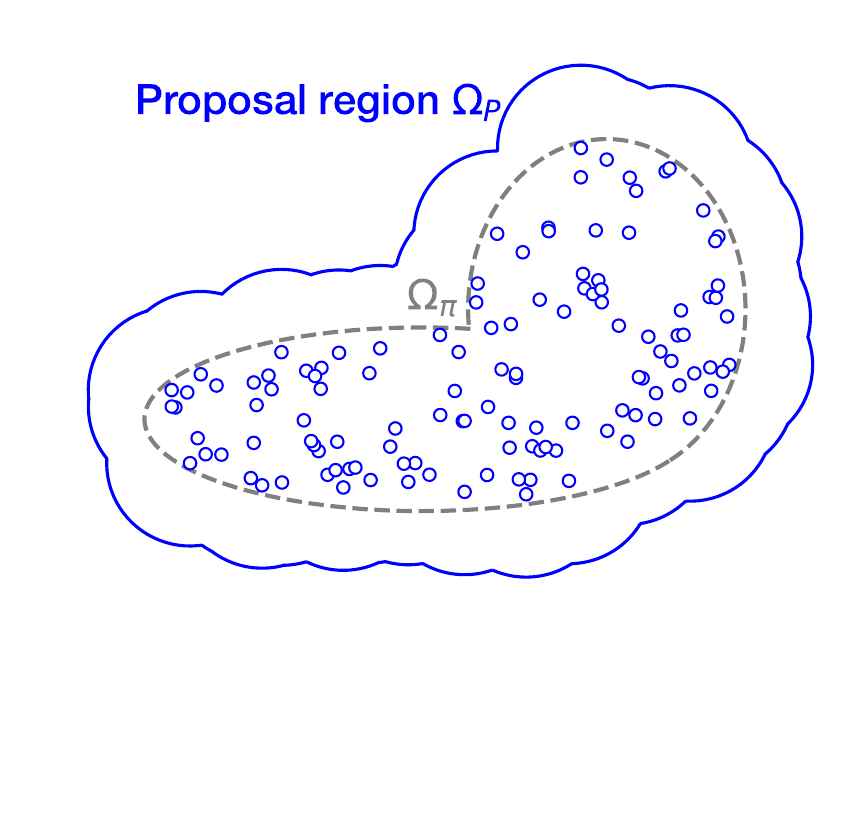}
\caption{Left panel: Illustration of the bagging procedure. The live points were simulated from the likelihood restricted prior, from the set $\Omega_\pi$. Sampling the live points with replacement yields a training sample (blue). The left-out points (orange) form the validation set. Right panel: The constructed proposal region $\Omega_P$ (blue outer contour) must be a superset of $\Omega_\pi$ for faithful nested sampling.}
\end{figure}

A safe cut-off distance $r_{\mathrm{train}}$ is the smallest $r$
such that every validation point lies within $\Omega^{j}(r)$ for
at least one training point $j\in{\cal I}_{\mathrm{train}}^{\mathrm{uniq}}$:
\begin{equation}
r_{\mathrm{train}}\!\left({\cal I}_{\mathrm{train}}\right)=
\begin{cases}
\displaystyle\max_{i\in{\cal I}\setminus{\cal I}_{\mathrm{train}}^{\mathrm{uniq}}}
\min_{j\in{\cal I}_{\mathrm{train}}^{\mathrm{uniq}}}
d\!\left(\theta^{i},\theta^{j}\right)
& \text{if } {\cal I}_{\mathrm{val}}\neq\emptyset,\\
0 & \text{if } {\cal I}_{\mathrm{val}}=\emptyset.
\end{cases}
\end{equation}
The final radius is the maximum over $m$ independent bootstrap
realisations:
\begin{equation}
r_{\mathrm{max}}=\max_{l=1,\ldots,m}r_{\mathrm{train}}^{l},
\end{equation}
where each $r_{\mathrm{train}}^{l}$ is computed from an independently
drawn training sequence. The proposal region is then the union of
balls of radius $r_{\mathrm{max}}$ centred at each live point:
\begin{equation}
\Omega_{P}=\left\{\theta\in\Omega_{\pi}:
\min_{j\in{\cal I}}d(\theta,\theta^{j})<r_{\mathrm{max}}\right\}.
\label{eq:OmegaP}
\end{equation}

\subsection{Binomial point process}
\label{subsec:ppp}

We model the $K$ live points as distributed uniformly, 
independently, and homogeneously (i.e., with equal density 
throughout) within $\Omega_{r\pi}$ with volume $V$. Under this model, the probability
that none of the $K$ points falls within Euclidean distance $r$ of
a fixed point $x\in\Omega_{r\pi}$ is exactly:
\begin{equation}
\Pr(\text{no point within }r)
=\left(1-\frac{V_d r^d}{V}\right)^{K},
\label{eq:binomial-cdf-main}
\end{equation}
where $V_{d}$ is the volume of the $d$-dimensional unit ball:
\begin{equation}
V_{d}=\frac{\pi^{d/2}}{\Gamma\!\left(\frac{d}{2}+1\right)}.
\end{equation}
This is the exact nearest-neighbour CDF for a Binomial point process
of $K$ uniform points in volume $V$.

Equation~\ref{eq:binomial-cdf-main} is exact under
Assumption~\ref{ass:interior}.  Boundary effects are excluded by Assumption~\ref{ass:interior}, and the
multi-modal case is not treated analytically here; in practice it is handled
by the cluster detection of Section~\ref{sec:implement}; consequences
for correctness and efficiency are stated there.

The analysis below assumes $\Omega_{r\pi}$ satisfies
Assumption~\ref{ass:interior}.

\section{Sampling correctly}\label{sec:correctness}

In this section, we analyse the reliability of the MLFriends proposal
region. We begin with a few clarifications.

\begin{remark}[Interior-ball analysis]
\label{ass:interior}
For all radii $r$ relevant to the analysis (i.e.\ $r \le r_{\mathrm{max}}$), we treat every ball of radius $r$ centred at a point in the likelihood-restricted prior region set $\Omega_{r\pi}$ as lying entirely within $\Omega_{r\pi}$.
Boundary effects---where a ball of radius $r$ extends outside
$\Omega_{r\pi}$---reduce the effective number of live points within
the ball, which tends to \emph{increase} $r_{\mathrm{max}}$ and
hence enlarge $\Omega_P$. The interior-ball analysis therefore gives
a conservative bound on $P^{\mathrm{missed}}$ (an overestimate of
the uncovered fraction) and an optimistic bound on the acceptance
rate (an underestimate of the true rate); these directions are
consistent throughout and stated without repetition hereafter.
\end{remark}

\begin{remark}[Multi-modality]
The above does not require a connected set. The multi-modal case and the ability to detect disconnected sets (clusters) is described in Section~\ref{sec:implement}.
\end{remark}

\begin{assumption}[Independence of nearest-neighbour distances]
\label{ass:indep}
The nearest-neighbour distances from the shared training set to each
validation point are treated as mutually independent.  Because all
distances are computed from the same $\bar{k}$ training points, this
is an approximation whose error is not characterised here.  All
results labelled \emph{Heuristic} depend on this assumption; a fully
rigorous treatment is left for future work.
\end{assumption}

Because the analysis rests on Assumptions~\ref{ass:interior}
and~\ref{ass:indep}, and on the approximate independence of bootstrap
rounds (Heuristic~\ref{le2}), results that depend on these
non-rigorous steps are labelled \emph{Heuristic} rather than
\emph{Theorem}.  A fully rigorous treatment, with explicit error
bounds as functions of $\epsilon$, $m$, $K$ and the geometry of
$\Omega_{r\pi}$, is left for future work.

\begin{lemma}\label{le0}
Let $\bar{k}=|{\cal I}_{\mathrm{train}}^{\mathrm{uniq}}|$ and
$n_v=|{\cal I}_{\mathrm{val}}|$ be fixed. Under the assumption that
the $K$ live points are distributed uniformly and independently in
$\Omega_{r\pi}$, and that the nearest-neighbour distances from the
training set to each validation point are independent, the expected $d$-th power of the bootstrap radius is:
\begin{equation}
E\!\left[r_{\mathrm{train}}^{d}\right]
= \frac{V}{V_d}\left(1 -
\frac{\Gamma(1+1/\bar{k})\,\Gamma(n_v+1)}
     {\Gamma(n_v+1+1/\bar{k})}\right).
\label{eq:rtrain-exact}
\end{equation}
For $\bar{k}\gg 1$ and $\ln(n_v+1)/\bar{k}\ll 1$, this simplifies to:
\begin{equation}
E\!\left[r_{\mathrm{train}}^{d}\right]
\approx \frac{\ln(n_v+1)}{\bar{k}}\times\frac{V}{V_d}.
\label{eq:rtrain-approx}
\end{equation}

The lemma is exact under Assumption~\ref{ass:indep} as stated in Section~\ref{sec:correctness}.
Relaxing Assumption~\ref{ass:indep} would require the joint
distribution of nearest-neighbour distances from a shared training
set to multiple validation points, which is left for future work.
\end{lemma}

\begin{proof}
Since the $K$ live points are uniform and independent in $\Omega_{r\pi}$
with volume $V$, the probability that none of the $\bar{k}$ training
points falls within Euclidean distance $r$ of a fixed validation
point is exactly:
\begin{equation}
\Pr(\text{no training point within }r)
=\left(1-\frac{V_d r^d}{V}\right)^{\bar{k}},
\label{eq:binomial-cdf}
\end{equation}
provided $V_d r^d \leq V$. This is the nearest-neighbour CDF for a Binomial point process of
$\bar{k}$ uniform points in volume $V$, exact under
Assumption~\ref{ass:interior}.
Inverting, we find $r^d = \frac{V}{V_d}(1-(1-u)^{1/\bar{k}})$ where
$u\sim\mathrm{Uniform}(0,1)$.

The training radius $r_{\mathrm{train}}$ is the maximum
nearest-neighbour distance over all $n_v$ validation points.
Under the assumption that these distances are independent,
$u_{\mathrm{max}}\sim\mathrm{Beta}(n_v,1)$ with density
$f(u)=n_v u^{n_v-1}$ on $[0,1]$. The expectation is:
\begin{equation}
E\!\left[r_{\mathrm{train}}^d\right]
= \frac{V}{V_d}\left(1 - E\!\left[(1-u_{\mathrm{max}})^{1/\bar{k}}\right]\right).
\end{equation}
Since $1-u_{\mathrm{max}}$ has density $n_v(1-w)^{n_v-1}$ on $[0,1]$:
\begin{equation}
E\!\left[(1-u_{\mathrm{max}})^{1/\bar{k}}\right]
= \int_0^1 w^{1/\bar{k}}\,n_v(1-w)^{n_v-1}\,dw
= n_v\,B\!\left(1+\tfrac{1}{\bar{k}},\,n_v\right)
= \frac{\Gamma(1+1/\bar{k})\,\Gamma(n_v+1)}{\Gamma(n_v+1+1/\bar{k})},
\end{equation}
where we used $n_v\,B(1+1/\bar{k},n_v) = n_v\,\Gamma(1+1/\bar{k})\Gamma(n_v)/\Gamma(n_v+1+1/\bar{k}) = \Gamma(1+1/\bar{k})\Gamma(n_v+1)/\Gamma(n_v+1+1/\bar{k})$.
This establishes eq.~\ref{eq:rtrain-exact}.

Although $r_{\mathrm{max}}$ is a random variable, substituting
$E[r_{\mathrm{max}}^d]$ into the expression for
$\Pr(\theta\notin\Omega_P)$ is not merely an approximation: since
$x\mapsto(1-V_d x/V)^K$ is convex in $x$ for $x\in[0,V/V_d]$
and $K\geq1$, Jensen's inequality gives
\begin{equation}
P^{\mathrm{missed}}
= E\!\left[\left(1-\frac{V_d r_{\mathrm{max}}^d}{V}\right)^K\right]
\;\leq\;
\left(1-\frac{V_d\,E[r_{\mathrm{max}}^d]}{V}\right)^K,
\end{equation}
so eq.~\ref{eq:Pmissed-exact} is a rigorous upper bound on
$P^{\mathrm{missed}}$ given the value of $E[r_{\mathrm{max}}^d]$
from Heuristic~\ref{le2}.  The plug-in step therefore strengthens
rather than weakens the result: the true uncovered fraction is at
most the quantity we compute.

For the large-$K$ approximation, note that $\Gamma(1+1/\bar{k})\to 1$
as $\bar{k}\to\infty$, and using the asymptotic expansion
$\ln\Gamma(n+a)-\ln\Gamma(n)\approx a\ln n$ for large $n$:
\begin{equation}
\frac{\Gamma(n_v+1)}{\Gamma(n_v+1+1/\bar{k})}
\approx (n_v+1)^{-1/\bar{k}},
\end{equation}
so that:
\begin{equation}
E\!\left[r_{\mathrm{train}}^d\right]
\approx \frac{V}{V_d}\left(1-(n_v+1)^{-1/\bar{k}}\right)
\approx \frac{V}{V_d}\cdot\frac{\ln(n_v+1)}{\bar{k}},
\end{equation}
where the approximations hold for $\bar{k}\gg 1$ and 
$\ln(n_v+1)/\bar{k}\ll 1$, both satisfied for practical $K$.
The last step uses $1-e^{-x}\approx x$ for small
$x=\ln(n_v+1)/\bar{k}$, valid when $n_v\ll e^{\bar{k}}$.

\end{proof}

\begin{lemma}\label{le1}
The number of unique indices $|{\cal I}_{\mathrm{train}}^{\mathrm{uniq}}|=s$ 
in a bootstrap sample of size $K$ drawn with replacement from 
$\{1,\ldots,K\}$ follows the distribution:
\begin{equation}
\Pr(|{\cal I}_{\mathrm{train}}^{\mathrm{uniq}}|=s)
=\frac{S_2(K,s)\,K!}{K^{K}\,(K-s)!},
\end{equation}
where $S_2(K,s)$ is the Stirling number of the second kind.
The expected number of unique indices is:
\begin{equation}
E\!\left[|{\cal I}_{\mathrm{train}}^{\mathrm{uniq}}|\right]
=K\!\left(1-\left(1-\tfrac{1}{K}\right)^{K}\right).
\end{equation}
The variance is:
\begin{equation}
\mathrm{Var}\!\left(|{\cal I}_{\mathrm{train}}^{\mathrm{uniq}}|\right)
=K\!\left(1-\tfrac{1}{K}\right)^{K}
+K(K-1)\!\left(1-\tfrac{2}{K}\right)^{K}
-K^{2}\!\left(1-\tfrac{1}{K}\right)^{2K}.
\end{equation}
\end{lemma}

\begin{proof}
The distribution of $|{\cal I}_{\mathrm{train}}^{\mathrm{uniq}}|$ 
is an instance of the classical occupancy problem, in which $K$ 
balls are thrown independently and uniformly into $K$ bins.
The number of ways to choose which $s$ distinct indices appear is 
$\binom{K}{s}$. The number of ways to assign $K$ draws to exactly 
those $s$ indices such that each appears at least once is 
$s!\,S_2(K,s)$, the number of surjections from $K$ draws onto $s$ 
labelled bins. Dividing by the total number of samples $K^K$ gives 
the stated distribution. The mean and variance follow from the 
indicator representation: letting 
$Y_j=\mathbf{1}(j\notin{\cal I}_{\mathrm{train}}^{\mathrm{uniq}})$,
each $Y_j\sim\mathrm{Bernoulli}((1-1/K)^K)$ and for $j\neq j'$,
$\Pr(Y_j=1,Y_{j'}=1)=(1-2/K)^K$, giving the stated mean and 
variance via $|{\cal I}_{\mathrm{train}}^{\mathrm{uniq}}|=K-\sum_j Y_j$
and the formula for the variance of a sum of exchangeable 
Bernoulli random variables.
\end{proof}

\begin{remark}
Since $(1-1/K)^{K}\to e^{-1}$ as $K\to\infty$:
\begin{equation}
\lim_{K\to\infty}\frac{E\!\left[|{\cal I}_{\mathrm{train}}^{\mathrm{uniq}}|\right]}{K}
=1-e^{-1}\approx0.6321,
\qquad
\lim_{K\to\infty}\frac{\mathrm{Var}\!\left(|{\cal I}_{\mathrm{train}}^{\mathrm{uniq}}|\right)}{K}
=e^{-1}-3e^{-2}+2e^{-3}\approx 0.105.
\end{equation}
The standard deviation of $|{\cal I}_{\mathrm{train}}^{\mathrm{uniq}}|$
is therefore $O(\sqrt{K})$, so the coefficient of variation
$\mathrm{std}/\mathrm{mean} = O(K^{-1/2})\to 0$. Consequently,
$|{\cal I}_{\mathrm{train}}^{\mathrm{uniq}}|\approx\frac{2}{3}K$ and
$|{\cal I}_{\mathrm{val}}|\approx\frac{1}{3}K$ concentrate around their exact means for large $K$, with relative fluctuations
of order $O(K^{-1/2})$. In Heuristic~\ref{le2} we additionally replace the
asymptotic constants $1-e^{-1}\approx 0.632$ and $e^{-1}\approx 0.368$ by
the simpler values $2/3$ and $1/3$, introducing a separate small constant
approximation error.
\end{remark}

\begin{heuristic}\label{le2}
The expected MLFriends radius for $K\gg3$ is:
\begin{equation}
E\!\left[r_{\mathrm{max}}^{d}\right]\approx
\frac{\ln\!\left(\frac{1}{3}Km\right)}{\frac{2}{3}K}
\times\frac{V}{V_{d}}.
\label{eq:rmax-approx}
\end{equation}
The independence of bootstrap rounds is an approximation.
Intuitively, a large training radius in one round (indicating sparse
coverage) suggests large radii in other rounds, because all rounds
draw from the same $K$ live points.  This positive dependence
suggests that the true distribution of $r_{\mathrm{max}}$ may be
stochastically smaller than under independence, so the independence
approximation may overestimate $r_{\mathrm{max}}$ and hence give a
larger proposal region.  If this is the case, the resulting bound on
$P^{\mathrm{missed}}$ is conservative (an overestimate), consistent
with Assumption~\ref{ass:interior}.  We do not prove this stochastic
ordering; it is noted as a plausibility argument and left for future
work.
\end{heuristic}

\begin{proof}[Derivation of Heuristic~\ref{le2}]
We use the exact expected training set size $\bar{k}$ and validation set size $n_v$ from Lemma~\ref{le1}:
\begin{equation}
\bar{k} = K\!\left(1-\left(1-\tfrac{1}{K}\right)^K\right),
\qquad
n_v = K\!\left(1-\tfrac{1}{K}\right)^K.
\end{equation}
Since $|{\cal I}_{\mathrm{train}}^{\mathrm{uniq}}|$ and
$|{\cal I}_{\mathrm{val}}|$ concentrate around $\bar{k}$ and $n_v$
with coefficient of variation $O(K^{-1/2})$ (Lemma~\ref{le1},
Remark), we substitute these means into the exact expression of
Lemma~\ref{le0} (eq.~\ref{eq:rtrain-exact}):
\begin{equation}
E\!\left[r_{\mathrm{train}}^{d}\right]
\approx \frac{V}{V_d}\left(1 -
\frac{\Gamma(1+1/\bar{k})\,\Gamma(n_v+1)}
     {\Gamma(n_v+1+1/\bar{k})}\right).
\end{equation}
For the maximum over $m$ independent rounds, under the independence
approximation the $m\,n_v$ validation distances are i.i.d., so
$u_{\mathrm{max}}\sim\mathrm{Beta}(m\,n_v,1)$. Applying
Lemma~\ref{le0} with effective validation size $m\,n_v$ gives:
\begin{equation}
E\!\left[r_{\mathrm{max}}^{d}\right]
= \frac{V}{V_d}\left(1 -
\frac{\Gamma(1+1/\bar{k})\,\Gamma(m\,n_v+1)}
     {\Gamma(m\,n_v+1+1/\bar{k})}\right).
\label{eq:rmax-exact}
\end{equation}

For the large-$K$ approximation, substitute
$\bar{k}\approx\frac{2}{3}K$ and $n_v\approx\frac{1}{3}K$,
use $\Gamma(1+1/\bar{k})\approx 1$ and
$\Gamma(m\,n_v+1)/\Gamma(m\,n_v+1+1/\bar{k})
\approx(m\,n_v+1)^{-1/\bar{k}}$,
and apply $1-(m\,n_v+1)^{-1/\bar{k}}
\approx\ln(m\,n_v+1)/\bar{k}
\approx\ln(\frac{1}{3}Km)/(\frac{2}{3}K)$
to obtain eq.~\ref{eq:rmax-approx}. Note that
$m n_v + 1 \approx \frac{1}{3}Km$ for large $K$.
\end{proof}

\begin{heuristic}\label{te1}
Under the assumption that live points are distributed uniformly
and independently in $\Omega_{r\pi}$, and using the radius
from Heuristic~\ref{le2} (exact expression eq.~\ref{eq:rmax-exact}),
an upper bound on the expected fraction of prior mass in $\Omega_{r\pi}$ not covered
by $\Omega_{P}$ is:
\begin{equation}
P^{\mathrm{missed}}
\;\leq\; \left(\frac{\Gamma(1+1/\bar{k})\,\Gamma(mn_v+1)}
             {\Gamma(mn_v+1+1/\bar{k})}\right)^K,
\label{eq:Pmissed-exact}
\end{equation}
where $\bar{k}=K(1-(1-1/K)^K)$ and $n_v=K(1-1/K)^K$.
For $K\gg3$, this simplifies to:
\begin{equation}
P^{\mathrm{missed}}\approx\left(\tfrac{1}{3}Km\right)^{-3/2}.
\label{eq:Pmissed-approx}
\end{equation}
This result is derived under Assumptions~\ref{ass:interior}
and~\ref{ass:indep}, and inherits the approximate independence of
bootstrap rounds from Heuristic~\ref{le2}.
\end{heuristic}

\begin{proof}[Derivation of Heuristic~\ref{te1}]
Since the $K$ live points are uniform and independent in
$\Omega_{r\pi}$ with volume $V$, and under
Assumption~\ref{ass:interior}, the probability that none falls
within distance $r$ of a fixed interior point
$\theta\in\Omega_{r\pi}$ is exactly:
\begin{equation}
\Pr(\theta\notin\Omega_{P}\mid r_{\mathrm{max}}=r)
=\left(1-\frac{V_d r^d}{V}\right)^K.
\end{equation}
The random variable $r_{\mathrm{max}}$ has expected $d$-th power
given by eq.~\ref{eq:rmax-exact}.  Because
$x\mapsto(1-V_d x/V)^K$ is convex in $x$, a direct plug-in of
$E[r_{\mathrm{max}}^d]$ overestimates $E[\Pr(\theta\notin\Omega_P)]$
by Jensen's inequality.  We therefore use the plug-in value as an
upper bound:
\begin{equation}
P^{\mathrm{missed}}
:= E\!\left[\frac{|\Omega_{r\pi}\setminus\Omega_P|}{|\Omega_{r\pi}|}\right]
\;\leq\;
\left(1-\frac{V_d\,E[r_{\mathrm{max}}^d]}{V}\right)^K.
\end{equation}
Substituting eq.~\ref{eq:rmax-exact} gives:
\begin{equation}
\frac{V_d\,E[r_{\mathrm{max}}^d]}{V}
= 1 - \frac{\Gamma(1+1/\bar{k})\,\Gamma(mn_v+1)}
           {\Gamma(mn_v+1+1/\bar{k})},
\end{equation}
and therefore:
\begin{equation}
\Pr(\theta\notin\Omega_{P})
=\left(\frac{\Gamma(1+1/\bar{k})\,\Gamma(mn_v+1)}
            {\Gamma(mn_v+1+1/\bar{k})}\right)^K.
\end{equation}
Under Assumption~\ref{ass:interior}, this probability is the same
for every interior point $\theta\in\Omega_{r\pi}$.
By Assumption~\ref{ass:interior}, boundary effects increase
$r_{\mathrm{max}}$ and reduce the uncovered fraction relative to
the interior bound, so the bound is conservative globally.
Integrating over $\Omega_{r\pi}$ with uniform prior density
establishes eq.~\ref{eq:Pmissed-exact} as an upper bound on
$P^{\mathrm{missed}}$.

For the large-$K$ approximation, the large-$K$ limit of
Heuristic~\ref{le2} gives
$V_d r_{\mathrm{max}}^d / V \approx
\ln(\frac{1}{3}Km)/(\frac{2}{3}K)$, so:
\begin{equation}
\Pr(\theta\notin\Omega_{P})
\approx\left(1-\frac{\ln(\frac{1}{3}Km)}{\frac{2}{3}K}\right)^K
\approx\exp\!\left(-\frac{3}{2}\ln\!\left(\tfrac{1}{3}Km\right)\right)
=\left(\tfrac{1}{3}Km\right)^{-3/2},
\end{equation}
where the second approximation uses $(1-x/K)^K\approx e^{-x}$
for large $K$, with $x=\frac{3}{2}\ln(\frac{1}{3}Km)$.
\end{proof}

With the approximations above, $P^{\mathrm{missed}}\approx(\frac{1}{3}Km)^{-3/2}$
is below $10^{-3}$ for $Km\gtrsim 300$.
The heuristic bound on $P^{\mathrm{missed}}$ depends only on $K$
and $m$, not directly on the dimensionality $d$; however,
the underlying independence assumptions become less accurate in
higher dimensions where boundary effects are more pronounced.

\citet{Buchner2014stats} proposed setting the number of bootstrap
rounds $m$ through an acceptable probability $\epsilon$ that a live
point was never in the validation set. They derived a corresponding
conservative $m$:
\begin{equation}
m=\left\lfloor\frac{\ln\epsilon-\ln K}
{\ln\!\left(1-\left(1-\frac{1}{K}\right)^{K}\right)}\right\rfloor.
\end{equation}
Here $\epsilon$ can be set to $\epsilon=1/N$ with $N$ the number of
times the LRPS is expected to construct the region. Practical
implementations typically use $m$ of the order of 20. 
When using $\epsilon=10^{-6}$ and $K=1000$ (giving $m\approx20$), we find
$P^{\mathrm{missed}}\approx2\times10^{-7}$, indicating that the
formula for $m$ given $\epsilon$ is indeed slightly more conservative.

\begin{proposition}\label{prop:volume}
Under the approximations of Lemma~\ref{le0} and Heuristics~\ref{le2}--\ref{te1},
the cumulative log-volume bias introduced by MLFriends's finite proposal region
over $N$ nested sampling iterations is bounded above by
$N\left(\tfrac{1}{3}Km\right)^{-3/2}$,
which is negligible compared to the inherent statistical standard deviation
$\sqrt{N/K}$ of the log-volume estimate, provided
$K^{1/2}\,m^{3/2} \gg \sqrt{N}\cdot 3^{3/2}$.
\end{proposition}

\begin{proof}[Proof of Proposition~\ref{prop:volume}]
We proceed by induction on the nested sampling iteration.

\textbf{Base case.}
In the first iteration, live points are sampled directly from the prior,
so the population is faithful by assumption: the live points are distributed
as $\mathrm{Uniform}(\Omega_{r\pi})$.

\textbf{Inductive step.}
Suppose the current live point population is faithful, i.e., distributed
as $\mathrm{Uniform}(\Omega_{r\pi})$.
MLFriends constructs $\Omega_P$ from this population.
Appendix~\ref{subsec:correctness-sampling} establishes that sampling
from $\Omega_P$ is exactly uniform over $\Omega_P$.

By Heuristic~\ref{te1}, the expected fraction of prior mass in $\Omega_{r\pi}$
not covered by $\Omega_P$ satisfies the upper bound:
\begin{equation}
P^{\mathrm{missed}}
:= E\!\left[\frac{|\Omega_{r\pi}\setminus\Omega_P|}{|\Omega_{r\pi}|}\right]
\;\leq\; \left(\tfrac{1}{3}Km\right)^{-3/2}.
\end{equation}
When $\Omega_P\supseteq\Omega_{r\pi}$, rejection sampling from
$\mathrm{Uniform}(\Omega_P)$ produces a point uniform on $\Omega_{r\pi}$,
and the population remains faithful.
When $\Omega_P\not\supseteq\Omega_{r\pi}$, accepted points are
uniform on $\Omega_P\cap\Omega_{r\pi}$ only.
Since the prior is uniform on $\Omega_{r\pi}$, the missed probability
mass fraction equals the missed volume fraction exactly:
\begin{equation}
\Pr\!\left(\theta\in\Omega_{r\pi}\setminus\Omega_P\right)
= \frac{|\Omega_{r\pi}\setminus\Omega_P|}{|\Omega_{r\pi}|},
\end{equation}
so the fractional volume bias per iteration is bounded above by
$P^{\mathrm{missed}} \leq \left(\tfrac{1}{3}Km\right)^{-3/2}$.
Within the present heuristic approximation, this bound is the same
at every iteration since $\left(\tfrac{1}{3}Km\right)^{-3/2}$ depends
only on the fixed parameters $K$ and $m$.

Accumulating over $N$ iterations, the total log-volume bias satisfies:
\begin{equation}
\left|\sum_{i=1}^{N}\ln\!\left(1-P^{\mathrm{missed}}\right)\right|
\leq N\cdot P^{\mathrm{missed}}
\leq N\left(\tfrac{1}{3}Km\right)^{-3/2},
\end{equation}
where the first inequality uses $-\ln(1-x) \leq x/(1-x) \approx x$
for $x \ll 1$, which holds since $P^{\mathrm{missed}} \ll 1$
in the regime of interest, and the second applies the bound from Heuristic~\ref{te1}.

The inherent statistical standard deviation of the log-volume estimate
after $N$ iterations is $\sqrt{N/K}$ \citep{skilling2009nested}.
For the log-volume bias to be negligible relative to this standard
deviation, we require:
\begin{equation}
N\left(\tfrac{1}{3}Km\right)^{-3/2} \ll \sqrt{N/K},
\quad\text{i.e.,}\quad
\sqrt{N}\,K^{1/2}\,m^{3/2} \gg 3^{3/2} \approx 5.2,
\end{equation}
which is satisfied for the moderate-to-large values of $K$ and $m$
used in practical nested sampling applications.
Therefore, within the present heuristic approximation, the log-volume
bias from MLFriends is negligible compared to the inherent variance
of nested sampling in practical regimes.
This is consistent with extending the convergence framework of
\citet{Chopin2010} and \citet{Salomone2018} to a full nested sampling algorithm including the
LRPS treatment, albeit at a heuristic level and leaving a
fully rigorous treatment as an open problem.
\end{proof}

\begin{assumption}[Likelihood homogeneity of missed regions]
\label{ass:likelihood-homogeneity}
The likelihood in $\Omega_{r\pi}\setminus\Omega_P$ is not
systematically higher than the average likelihood on $\Omega_{r\pi}$.
This is plausible because missed regions are furthest from any live point
and hence typically in low-posterior-density areas.
A pathological violation would be a tiny isolated region of
extremely high likelihood that is missed entirely.
Such a failure mode is not specific to MLFriends: any Monte Carlo method
with a finite sample can miss an isolated high-likelihood region.
For nested sampling specifically, if the likelihood increases
continuously towards such a region, the live points nearest to it
will have above-average likelihoods and will therefore be replaced
last; the region is then sampled in subsequent iterations as the
likelihood threshold rises.
This argument requires the likelihood to be sufficiently regular
(no discontinuous jumps), which is assumed here.
\end{assumption}

\begin{corollary}\label{corr1}
Under the approximations of Proposition~\ref{prop:volume} and
Assumption~\ref{ass:likelihood-homogeneity}, the fractional bias in
the marginal likelihood estimate $\hat{Z}$ introduced by MLFriends's
finite proposal region satisfies:
\begin{equation}
\frac{Z - \hat{Z}}{Z}
\;\leq\; N \left(\tfrac{1}{3}Km\right)^{-3/2},
\end{equation}
which is negligible compared to the inherent statistical uncertainty
$\sqrt{N/K}$ of $\hat{Z}/Z$, provided
$K^{1/2}\,m^{3/2} \gg \sqrt{N}\cdot 3^{3/2}$.
For example, this holds for $K=400$, $m=20$, and $N\leq 10^5$,
covering the vast majority of practical nested sampling applications.
\end{corollary}

\begin{proof}[Proof of Corollary~\ref{corr1}]
The marginal likelihood estimate is:
\begin{equation}
\hat{Z} = \sum_{i=1}^{N} {\cal L}_i \, \Delta V_i,
\end{equation}
where $\Delta V_i$ is the prior volume assigned to iteration $i$.
The volume increments $\Delta V_i$ are computed solely from the
shrinkage formula $(K-1)/K$ per iteration and are independent of
whether $\Omega_P \supseteq \Omega_{r\pi}$; they are therefore the
same as in a faithful run (see Remark~\ref{rem:volume-increments}).

If the proposal region misses a fraction $P^{\mathrm{missed}}$
of $\Omega_{r\pi}$, accepted points are drawn only from
$\Omega_P \cap \Omega_{r\pi}$.
The fractional contribution of the missed region to the true evidence
at iteration $i$ is:
\begin{equation}
\frac{\int_{\Omega_{r\pi}\setminus\Omega_P}
      {\cal L}(\theta)\,d\pi(\theta)}
     {\int_{\Omega_{r\pi}}{\cal L}(\theta)\,d\pi(\theta)},
\end{equation}
which equals $P^{\mathrm{missed}}$ only if ${\cal L}$ is constant on
$\Omega_{r\pi}$, and can exceed it if the missed region is enriched
in likelihood.
Under Assumption~\ref{ass:likelihood-homogeneity}, however, the
missed region carries no more than its proportional share of
likelihood, so this ratio is bounded above by $P^{\mathrm{missed}}$.
The missed prior mass in $\Omega_{r\pi}\setminus\Omega_P$ is never
sampled and never contributes to $\hat{Z}$, so $\hat{Z}$ is biased
downward.

Accumulating the fractional downward bias over $N$ iterations:
\begin{equation}
\frac{Z - \hat{Z}}{Z}
\;\leq\; 1 - \left(1 - P^{\mathrm{missed}}\right)^N
\;\leq\; N \cdot P^{\mathrm{missed}}
\;\leq\; N \left(\tfrac{1}{3}Km\right)^{-3/2},
\end{equation}
where the second inequality uses $1-(1-p)^N \leq Np$ for $p\in[0,1]$,
and the third applies the upper bound on $P^{\mathrm{missed}}$
from Heuristic~\ref{te1}.

The statistical standard deviation of $\hat{Z}/Z$ from the inherent
volume-shrinkage uncertainty is of order $\sqrt{N/K}$
\citep{skilling2009nested}.
The fractional bias in $\hat{Z}$ is therefore negligible relative to
its statistical uncertainty when
$\sqrt{N/K} \gg N\left(\tfrac{1}{3}Km\right)^{-3/2}$, i.e., when
$K^{1/2}\,m^{3/2} \gg \sqrt{N}\cdot 3^{3/2}$,
which holds for the moderate-to-large values of $K$ and $m$
used in practice; for example, $K=400$, $m=20$, and $N\leq 10^5$.
\end{proof}

\begin{remark}[Volume increments are not inflated]
\label{rem:volume-increments}
One might expect that missing part of $\Omega_{r\pi}$ causes the
nested sampling volume increments $\Delta V_i$ to be overestimated,
and hence $\hat{Z}$ to be upward-biased.  This is not the case.
The volume increments are computed solely from the shrinkage formula
$(K-1)/K$ per iteration and are independent of whether
$\Omega_P\supseteq\Omega_{r\pi}$.  They are therefore the same as
in a faithful run, and their sum still approximates 1.  The
downward bias in $\hat{Z}$ arises instead because the
likelihood-weighted prior mass in $\Omega_{r\pi}\setminus\Omega_P$
is never sampled and never enters the sum
$\hat{Z}=\sum_i{\cal L}_i\,\Delta V_i$: the volume increments are
assigned to the wrong (accessible) region, leaving the missed
region unaccounted for.  The two effects---faster advance of
$L_{\min}$ and underestimation of $Z$---are the same phenomenon
viewed from different perspectives.
\end{remark}

\begin{remark}[Effect on the posterior samples]
The normalised posterior weights $w_i/\sum_i w_i$ are unaffected
by a global rescaling of all volume increments $\Delta V_i$, so
the posterior approximation is insensitive to any uniform volume
bias.  The dominant source of posterior error from MLFriends is
instead the missed prior mass: points in
$\Omega_{r\pi}\setminus\Omega_P$ are never proposed and receive
zero weight, so the corresponding posterior mass is absent from
the approximation entirely.

Under Assumption~\ref{ass:likelihood-homogeneity}, 
the missed region carries at most a
fraction $P^{\mathrm{missed}}$ of the true posterior mass
(since posterior mass is proportional to likelihood times prior
volume, and neither is elevated in the missed region by
assumption).  The posterior approximation therefore omits at
most a fraction $P^{\mathrm{missed}}$ of the total posterior
mass, which is negligible for the parameter choices discussed
in Corollary~\ref{corr1}.

If the likelihood-homogeneity assumption fails---for example,
if a small isolated high-likelihood mode lies entirely outside
$\Omega_P$---the missed posterior mass could be much larger than
$P^{\mathrm{missed}}$ suggests.  The cluster detection described
in Section~\ref{sec:implement} is specifically designed to detect
such isolated modes and include them in $\Omega_P$.
\end{remark}

\section{Sampling efficiently}\label{sec:efficiency}

In this section, we comment on the sampling efficiency of MLFriends.
Implementation details and the correctness of sampling from the 
restricted prior $\Omega_P$ are given in Appendix~\ref{sec:implement}.
For region-based samplers, the acceptance rate is:
\begin{equation}
A = \frac{|\Omega_{r\pi}|}{|\Omega_{P}|},
\end{equation}.
This is because $\Omega_P \supseteq \Omega_{r\pi}$ is ensured with high
probability by Heuristic~\ref{te1}.

As an illustrative example, consider a $d$-dimensional box
$\Omega_{r\pi}=[0,a]^{d}$ with volume $V_{r\pi}=a^{d}$. The
proposal region $\Omega_{P}$ is contained within
$[-r_{\mathrm{max}},a+r_{\mathrm{max}}]^{d}$, giving:
\begin{equation}
A\geq\frac{a^{d}}{(a+2r_{\mathrm{max}})^{d}}
=\left(1+\frac{2r_{\mathrm{max}}}{a}\right)^{-d}.
\end{equation}
Setting w.l.o.g. $a=1$, $V=1$ and $m=20$, and substituting the
large-$K$ approximation eq.~\ref{eq:rmax-approx}:
\begin{equation}
A\geq\left(1+2\left(
\frac{\ln(\frac{1}{3}Km)}{\frac{2}{3}K}
\times\frac{1}{V_{d}}\right)^{1/d}\right)^{-d}.
\end{equation}
With $K=1000$ and $m=20$, this lower bound evaluates to approximately
$70\%$ at $d=2$, $7\%$ at $d=5$, and $0.00026\%$ at $d=10$,
illustrating the severe curse of dimensionality. This was already
noted in \citet{Buchner2014stats} and is discussed further in
\citet{Buchner2021c}.

This efficiency analysis rests on Assumption~\ref{ass:indep} from Section~\ref{sec:correctness}, and on the
approximate independence of bootstrap rounds (Heuristic~\ref{le2}).
Additionally, it assumes $\Omega_{r\pi}$ is a convex box; under
Assumption~\ref{ass:interior}, the acceptance rate bound is
optimistic near boundaries and in disjoint (multi-modal) scenarios.
A full treatment of these effects is left for future work.

\section*{Connection to snowballing nested sampling}

\citet{Buchner2023} proposed \emph{snowballing nested 
sampling}, in which nested sampling is run repeatedly with 
increasing numbers of live points $K$, while reusing previously 
computed likelihood evaluations. This avoids the typical 
problem of having to test and choose a sufficiently number of MCMC steps $M$ 
in some likelihood-restricted prior sampling algorithms (see \cite{Buchner2021c} for a review).
Increasing $K$ improves the quality of the likelihood-restricted 
prior sampling (LRPS) without requiring the user to know a sufficient 
number of MCMC steps $M$ in advance. While the convergence argument 
presented in that paper is informal and the stated equivalence 
between $K$ and $M$ is not derived rigorously, the underlying 
insight has merit and can be stated precisely in the present 
framework.

The key observation is that the quality of the LRPS at each 
iteration depends not on $M$ alone but on the product $K \times M$.
At each nested sampling iteration with $K$ live points, the 
likelihood threshold advances by a fractional prior volume of 
approximately $1/K$ (the mean of a $\mathrm{Beta}(1,K)$ 
distribution).  The MCMC chain of $M$ steps must produce a new 
point approximately distributed according to the new constrained 
prior $\eta(\,\cdot\,;L_{\min})$, starting from a seed that is 
distributed according to the previous constrained prior.  As $K$ 
increases, the constrained region changes more slowly between 
consecutive iterations, so the seed is closer in distribution to 
the target, and a shorter chain suffices for a given mixing 
quality.  For a fixed ergodic proposal kernel, the mixing 
requirement is approximately $M \times K \gg 1$, and the LRPS 
quality is approximately invariant under the rescaling $K \to cK$, 
$M \to M/c$ for fixed $c > 0$, provided both remain large.
This is consistent with the empirical observation of 
\citet{Buchner2023} that nested sampling convergence depends on 
$K \times M$ rather than on $K$ or $M$ individually.

Now let's consider the application of snowballing nested sampling 
to region-based LRPS. Heuristic~\ref{te1} shows that the 
MLFriends proposal region covers the constrained prior with missed 
fraction bounded above by $(\frac{1}{3}Km)^{-3/2}$, where $m$ is 
the number of bootstrap rounds. As $K$ increases in snowballing 
nested sampling, this bound decreases as $K^{-3/2}$, so the 
region-based LRPS becomes more faithful automatically, without any 
change to $m$. Although the number of iterations also grows as 
$N_K \approx KH$, the net effect is that the cumulative bias bound 
of Corollary~\ref{corr1} decays as $N_K \cdot (\frac{1}{3}Km)^{-3/2} 
\propto K^{-1/2} \to 0$ as $K\to\infty$, with $m$ and $H$ fixed. 
Snowballing nested sampling with MLFriends therefore achieves 
vanishing bias in the evidence estimate as $K\to\infty$, under the 
heuristic approximations of Corollary~\ref{corr1}. A rigorous 
convergence proof in the sense of \citet{Salomone2018} remains an 
open problem.

\section{Discussion}\label{sec:Discussion}

Today, overwhelming dataset sizes can prohibit visual inspection of
Bayesian inference results on a case-by-case basis. Robust algorithms
with analytical reliability estimates are therefore valuable, even
at some computational cost. This paper has presented an analytical
characterisation of MLFriends and argued heuristically
that its proposal region covers the
likelihood-restricted prior with high probability, making it
practically reliable for low-to-moderate dimensional inference.

MLFriends takes ideas from agglomerative clustering, graph theory,
bootstrapping, bagging and cross-validation. Its principles can be
combined with other region-based algorithms. For example, a robust
enlargement factor for ellipsoidal or multi-ellipsoidal nested
sampling could be determined with the bootstrapping method. Based on
this idea, the efficiency of ellipsoidal nested sampling was
numerically studied in \citet{Buchner2021c} as a function of
dimensionality and live points.

MLFriends has also been compared to obtaining the support of a
kernel density estimate with a top-hat kernel. The basic algorithm
was presented with $\ell^{1}$ and $\ell^{2}$ norms in
\citet{Buchner2014stats}. Subsequently, \citet{Buchner2019c} adopted
the Mahalanobis metric and proposed iterative metric learning by
cluster co-centering. This leads to an acceptably efficient algorithm
that is robust across a wide class of Bayesian inference problems
with 1--10 parameters. For these settings, MLFriends is perhaps the
simplest algorithm to implement from scratch. It has been adopted
into the general-purpose nested sampling packages \textsc{dynesty}
\citep{Speagle2020} and \textsc{UltraNest} \citep{UltraNest}.

The main analytical contribution of this work is a heuristic bound
on the probability that the MLFriends proposal region fails to cover
the likelihood-restricted prior. This probability is found to be
exceedingly small, providing the
first analytical reliability estimates for a fully specified,
practically implementable nested sampling algorithm. Our analysis
has several limitations. First, all results rest on approximating
the live-point process by a homogeneous Binomial process, which
is inaccurate near boundaries of $\Omega_{r\pi}$ (see Assumptions~\ref{ass:interior}).
Near boundaries, the effective density of live points within a
ball of fixed radius is lower, which tends to increase the
estimated radius $r_{\mathrm{max}}$ and thus enlarge $\Omega_P$.
This plausibly improves coverage, making our heuristic bounds conservative.
Conversely, it reduces the acceptance rate. The quantitative impact of both
effects likely grows with the dimension and the surface-to-volume ratio
of $\Omega_{r\pi}$.
Second, we treat the $m$ bootstrap rounds and nearest-neighbour
distances as independent (see Assumption~\ref{ass:indep} and
Heuristic~\ref{le2}); the resulting error is not characterised. 
Third, both Proposition~\ref{prop:volume} and Corollary~\ref{corr1}
are conditional on these heuristic bounds and do not constitute
rigorous convergence proofs.
A fully rigorous treatment, with explicit error bounds tracked
through all iterations as functions of $\epsilon$, $m$ and $K$, is
left for future work.

For high-dimensional inference, MLFriends becomes impractical due to
the curse of dimensionality, and MCMC-based step samplers are
preferred \citep{Jasa2012,Handley2015}. 
\citet{Salomone2018} establish consistency of a closely related
sequential Monte Carlo reformulation of nested sampling under an $\eta_t$-invariant
MCMC kernel, but leave open the consistency of the original
nested sampling algorithm under a specific LRPS implementation.
The present work addresses this gap for MLFriends, providing
quantitative bounds on the proposal coverage error.
While MLFriends is focused on low dimensions, in divide-and-conquer 
approaches employed in astrophysics
\citep[e.g.][]{Buchner2015,Baronchelli2018,Farr2019,Miller2020,Essick2022,Talbot2023},
sometimes high-dimensional hierarchical inference can be decomposed
into many low-dimensional subproblems. For these, MLFriends may be
the algorithm of choice due to its robustness to multi-modality and
the reliability properties illustrated here.

\backmatter

\bmhead{Acknowledgements}

JB thanks Radu Craiu, Ewan Cameron and Stefan Czesla for comments on the manuscript. JB thanks Max Isi, J Michael Burgess, Josh Speagle for insightful conversations. After an initial write-up, OpenAI ChatGPT and Anthropic Claude Sonnet 4.6 have been used to vet the arguments and improve the writing. The authors take responsibility for the full manuscript.

\bibliography{ref}

\begin{appendices}

\section{Implementation details for efficient sampling}\label{sec:implement}

The following sections present implementation details and efficiency
improvements for handling complicated geometries. The metric learning
described in Section~\ref{subsec:metric} determines the distance
metric used in the analysis of Sections~\ref{sec:correctness}
and~\ref{sec:efficiency}: after the affine transformation to
Mahalanobis coordinates, the analysis applies with the Euclidean
metric.

\subsection{Friends and clusters}

The proposal region $\Omega_{P}$ need not be convex or connected:
multiple `clusters' or `modes' are possible. To characterise this
structure, define the set of live points directly connected to live
point $i$ (its `friends'):
\begin{equation}
F_{1}^{i}=\left\{j\in{\cal I}:d(\theta^{i},\theta^{j})<r_{\mathrm{max}}\right\}.
\end{equation}
Cluster membership is then defined recursively by agglomerative
single-linkage clustering:
\begin{equation}
F_{l+1}^{i}=F_{l}^{i}\cup\left\{j\in{\cal I}\setminus F_{l}^{i}:
\min_{k\in F_{l}^{i}}d(\theta^{k},\theta^{j})<r_{\mathrm{max}}\right\}.
\end{equation}
We write $F^{i}=F_{\infty}^{i}$ for the full cluster containing
point $i$. Typically only a few recursions are needed because
$r_{\mathrm{max}}$ is constructed so that each live point can reach
several others.

The set of all clusters is
${\cal F}=\{F^{i}\}_{i\in{\cal I}}$ (taking unique sets), and the
number of clusters is $N_{\mathrm{clusters}}=|{\cal F}|$. Notably,
no tuning parameter controls the number of clusters; it emerges
naturally from the data and the radius $r_{\mathrm{max}}$.

For cluster $c$ with member set $F^{c}$, the cluster mean and
empirical covariance are:
\begin{equation}
\mu^{c}=\frac{1}{|F^{c}|}\sum_{j\in F^{c}}\theta^{j},
\qquad
S_{c}=\frac{1}{|F^{c}|}\sum_{j\in F^{c}}
(\theta^{j}-\mu^{c})(\theta^{j}-\mu^{c})^{\top}.
\end{equation}

\subsection{Metric learning}
\label{subsec:metric}

MLFriends uses the Mahalanobis distance:
\begin{equation}
d_{M}(a,b\mid S)=\sqrt{(a-b)^{\top}S^{-1}(a-b)},
\end{equation}
where $S$ is a positive definite scale matrix. In the first iteration,
the Euclidean metric is used: $S=I_{d}$.

A naive choice for subsequent iterations would be the empirical
covariance over all live points:
\begin{equation}
S'=\frac{1}{K}\sum_{i\in{\cal I}}(\theta^{i}-\mu)(\theta^{i}-\mu)^{\top},
\qquad
\mu=\frac{1}{K}\sum_{i\in{\cal I}}\theta^{i}.
\end{equation}
However, this is inefficient when the posterior is multi-modal: the
covariance is then dominated by the separation between cluster centres
rather than the shape of individual clusters, so the Mahalanobis
balls are poorly aligned with the local geometry.

A better choice is to co-centre the clusters before computing the
covariance:
\begin{equation}
S'=\frac{1}{\sum_{c}|F^{c}|}
\sum_{c}\sum_{j\in F^{c}}(\theta^{j}-\mu^{c})(\theta^{j}-\mu^{c})^{\top}.
\end{equation}
We set $S\leftarrow S'$ for the next iteration of MLFriends.

\subsection{Sampling new live points in practice}

Sampling uniformly from $\Omega_{P}$ can be achieved by either of
two strategies; both produce exact uniform samples by the arguments
given in the Appendix~\ref{subsec:correctness-sampling}.

\paragraph{Strategy 1: bounding box rejection sampling.}
Sample $\theta^{*}$ uniformly from a bounding hyper-box
$\Omega_{\mathrm{box}}$ and accept if $\theta^{*}\in\Omega_{P}$:
\begin{equation}
A_{\mathrm{proposal}}(\theta^{*})=\mathbf{1}(\theta^{*}\in\Omega_{P}).
\end{equation}
The box is defined by:
\begin{equation}
\Omega_{\mathrm{box}}=\{\theta\in\Omega_{\pi}:l_{i}\leq\theta_{i}\leq u_{i}\},
\end{equation}
with edges:
\begin{gather}
l_{i}=\max\!\left(0,\;\min_{j\in{\cal I}}\theta_{i}^{j}-r_{\mathrm{max}}\right),\\
u_{i}=\min\!\left(1,\;\max_{j\in{\cal I}}\theta_{i}^{j}+r_{\mathrm{max}}\right).
\end{gather}
Since the Euclidean distance is bounded above by the $\ell^{\infty}$
distance, any point within Euclidean distance $r_{\mathrm{max}}$ of
a live point also lies within $\ell^{\infty}$ distance
$r_{\mathrm{max}}$, so $\Omega_{P}\subseteq\Omega_{\mathrm{box}}$.
Sampling from $\Omega_{\mathrm{box}}$ is straightforward with uniform
pseudo-random number generators. The efficiency of this strategy is
$|\Omega_{P}|/|\Omega_{\mathrm{box}}\cap\Omega_{\pi}|$, which exceeds
$|\Omega_{P}|/|\Omega_{\pi}|$ for sampling from the prior directly.
It is near unity in the early nested sampling iterations but may
decrease at later iterations, especially when clusters are present.

\paragraph{Strategy 2: live-point neighbourhood sampling.}
Choose a live point uniformly at random,
$i^{*}\sim\mathrm{Uniform}\{1,\ldots,K\}$, and propose a point by
sampling uniformly from the full $d$-dimensional ball of radius
$r_{\mathrm{max}}$ centred at $\theta^{i^{*}}$. To do so, draw a direction
$\mathbf{v}\sim\mathrm{Normal}(\mathbf{0},I_{d})$ and a radius
fraction $u\sim\mathrm{Uniform}(0,1)$, then set:
\begin{equation}
\theta^{*}=\theta^{i^{*}}+\frac{\mathbf{v}}{|\mathbf{v}|}
\times r_{\mathrm{max}}\times u^{1/d}.
\end{equation}
If $\theta^{*}\notin\Omega_{\pi}$, reject and redraw. Otherwise,
because the neighbourhoods $\Omega^{i}(r_{\mathrm{max}})$ can overlap,
the overlap is corrected by accepting with probability:
\begin{equation}
A_{\mathrm{proposal}}(\theta^{*})=\frac{1}{n_{\mathrm{overlap}}},
\end{equation}
where
\begin{equation}
n_{\mathrm{overlap}}=\left|\left\{i\in{\cal I}:
d(\theta^{*},\theta^{i})<r_{\mathrm{max}}\right\}\right|
\end{equation}
is the number of live points whose neighbourhood contains $\theta^{*}$.
Note that $A_{\mathrm{proposal}}\geq1/K$. In practice, the
computational cost of MLFriends is typically negligible compared to
the cost of evaluating the model likelihood.

\subsection{Correctness of sampling from \texorpdfstring{$\Omega_P$}{OmegaP}}
\label{subsec:correctness-sampling}

The rejection sampling step requires drawing
$\theta^{*}\sim\mathrm{Uniform}(\Omega_{P})$. We verify that both
strategies above achieve this exactly.

\paragraph{Strategy 1.}
Since $\Omega_{P}\subseteq\Omega_{\mathrm{box}}$, sampling uniformly
from $\Omega_{\mathrm{box}}$ and accepting points in $\Omega_{P}$ is
a standard rejection sampler. The accepted points are exactly uniform
over $\Omega_{P}$.

\paragraph{Strategy 2.}
Since $i^{*}$ is chosen uniformly from $\{1,\ldots,K\}$, each point
$\theta^{*}\in\Omega_{P}$ can be proposed from exactly
$n_{\mathrm{overlap}}(\theta^{*})$ live-point neighbourhoods, each
chosen with probability $1/K$. 
Because each neighbourhood proposal is uniform on a ball of the same
radius $r_{\mathrm{max}}$, the proposal density at $\theta^{*}$ is
proportional to $n_{\mathrm{overlap}}(\theta^{*})/K$.
Accepting with probability $1/n_{\mathrm{overlap}}(\theta^{*})$
cancels the over-representation, yielding a density proportional
to $1/K$ independently of $\theta^{*}$, i.e.\ uniform over
$\Omega_{P}$. Proposals outside $\Omega_{\pi}$ are rejected and
redrawn; since $\Omega_P\subseteq\Omega_\pi$ by definition
(eq.~\ref{eq:OmegaP}), this rejection step does not alter the
uniform distribution over $\Omega_P$.
\end{appendices}

\end{document}